\documentclass{svproc}
\usepackage{float}
\usepackage{algorithm} 
\usepackage{graphicx}
\usepackage{amsmath,amssymb}
\usepackage{url}
\usepackage{algorithm}
\usepackage[noend]{algpseudocode}
\usepackage{hyperref}
\usepackage{xlop}
\usepackage{wrapfig,tabularx}
\usepackage{booktabs}
\usepackage{graphicx}
\usepackage{xcolor}
\usepackage {longtable}
\usepackage[english]{babel}
\usepackage{array}
\usepackage[utf8]{inputenc}
\usepackage{longtable}
\usepackage{booktabs}
\usepackage{multirow}
\usepackage{siunitx}
\usepackage{color, colortbl}
\usepackage{caption}
\usepackage{subcaption}

\usepackage{tikz}
\usetikzlibrary{positioning}
\usepackage{pgfplots}
\usepackage{microtype}
\usepackage{caption,subcaption}

\begin{document}
\mainmatter              

\title{Fair Allocation Based Soft Load Shedding}
\titlerunning{Fair allocation based soft load shedding}  
%
\author{Sarwan Ali \inst{1} \and
	Haris Mansoor\inst{1} \and
	Imdadullah Khan\inst{1} \and
	Naveed Arshad \inst{1} \and
Safiullah Faizullah \inst{2}  \and
Muhammad Asad Khan \inst{3}
}

\authorrunning{Ali  et al.} 

\institute{Department of Compute Science, Lahore University of Management Sciences (LUMS), Lahore, Pakistan \\
	\email{\{16030030,16060061,imdad.khan,naveedarshad\}@lums.edu.pk}
	\and
	Department of Compute Science, Islamic University, \\
	Madinah, Saudi Arabia \\
	\email{safi@iu.edu.sa}
	\and
	Department of Telecommunication, Hazara University, \\ 
	Mansehra, Pakistan \\
	\email{asadkhan@hu.edu.pk}
}

\maketitle              

\begin{abstract}
	Renewable sources are taking center stage in electricity generation. Due to the intermittent nature of these renewable resources, the problem of the demand-supply gap arises. To solve this problem, several techniques have been proposed in the literature in terms of cost (adding {\em peaker plants}), availability of data ({\em Demand Side Management} ``DSM"), hardware infrastructure ({\em appliance controlling} DSM) and safety ({\em voltage reduction}). However, these solutions are not fair in terms of electricity distribution. In many cases, although the available supply may not match the demand in peak hours, however, the total aggregated demand remains less than the total supply for the whole day. Load shedding (complete blackout) is a commonly used solution to deal with the demand-supply gap, which can cause substantial economic losses. To solve the demand-supply gap problem, we propose a solution called {\em Soft Load Shedding} (SLS), which assigns electricity quota to each household in a fair way. We measure the fairness of SLS by defining a function for household satisfaction level. We model the household utilities by parametric function and formulate the problem of SLS as a social welfare problem. We also consider revenue generated from the fair allocation as a performance measure. To evaluate our approach, extensive experiments have been performed on both synthetic and real-world datasets, and our model is compared with several baselines to show its effectiveness in terms of fair allocation and revenue generation.
	
\end{abstract}


\section{Introduction}
Many countries have plans to shift to renewable sources of electricity by 2050 \cite{cosic2011towards,article}. This is mainly due to environmental and cost-related problems with fossil fuel-based power plants. While renewable sources (such as solar and wind) provide an environmentally friendly solution to the energy demand, they also create challenges in the electricity distribution system. The problems mostly originate from the variable and stochastic nature of renewable sources, which introduces the demand-supply gap of electricity \cite{lusis2017short}.

Utility companies use different approaches to handle the demand-supply gap. The most commonly used method is to use fossil fuel-based peaker plants. However, these stand-by plants are incredibly costly to operate \cite{vardakas2015survey}. {\em Real-Time Pricing} (RTP) is another method to match demand with supply, where the price of electricity changes in real-time. At peak load hours, the price is high, which restrain consumers/households from using excessive electricity \cite{chen2010two,dutta2017literature}. However, real-time pricing schemes have a rebound effect when the price falls, which induces artificial peaks. Demand-side management (DSM) techniques have been proposed to reduce the total load or at least partially shift it to non-peak hours. In some DSM techniques, the utility companies must have direct control over consumers' appliances. This is not only difficult to practically implement but also illegal in some countries. Methods like electricity storage \cite{roberts2011role,mohd2008challenges}, where the utility maintains a reservoir of excess energy produced during non-peak hours for peak time usage, is costly. {\em Electricity curtailment} programs \cite{aalami2010demand}, that provide (financial) incentives to consumers to reduce electricity consumption does not provide any reduction guarantee from the consumers. Varying electricity voltage \cite{craciun2009new} also has severe limitations as it affects stability and may lead to failure of costly appliances. In the worst-case scenario, utility companies have to perform load shedding (complete blackout) in order to deal with the demand-supply gap.

Electricity loads are broadly of two types, flexible and inflexible. Flexible loads can be shifted to another time window e.g. some industrial units can be operated at different hours of the day. Inflexible loads, on the other hand, are those where the activity can only take place at a specific time. The load of residential consumers is semi-flexible, i.e. while some appliances are crucial to be run at a specific time, others can be scheduled at a different time \cite{aslam2018soft}. While exploiting this {\em demand elasticity}, we purpose a soft load shedding (\textsc{SLS}) approach based on allocating a limited quota of electricity to each consumer during peak hours. Thus, when the total demand exceeds the available supply, the \textsc{SLS} scheme will allow consumers to use only essential appliances and keep the total load within the supply limit. This will let consumers play an active role in managing demand by prioritizing their loads and exercising their maximum flexibility while retaining control of their appliances.

With {\em Advanced Metering Infrastructure} (AMI) that has the capability of {\em threshold metering}, implementation of a \textsc{SLS} scheme is feasible. These smart meters can be programmed remotely to limit supply within a fixed period to a certain quota \cite{dutta2017literature,vardakas2015survey}. Note that the implementation of \textsc{SLS} requires certain provisions in customers' contracts. In this paper, we assume that such provisions are available and focus only on the data analytics aspects of rolling out such a scheme. 

There are significant data analytics challenges for rolling out the \textsc{SLS} mechanism. Firstly, it requires knowledge of household demands and the available supply at a given time. Several predictive data analytics and forecasting methods have been proposed for these problems \cite{ali2019short,ali2019hour,shi2018deep,gerossier2017probabilistic}. Secondly, with the increasing number of households, obtaining an optimal solution becomes computationally infeasible. One way to deal with this problem is to cluster the consumers based on their demands and use the aggregated demand of a set of households rather than the individual household. Another way is to aggregate the demand of a day ($24$ hours) or week and assign the quota to cluster of household for the whole day or week rather than an individual hour. To cluster the households and aggregate the daily/weekly load, several techniques have been proposed in literature \cite{kell2018segmenting,ali2019short,ali2019hour}.
In this paper, since the number of consumers is small for each dataset, we are not clustering the consumers.

The goal of \textsc{SLS} is to assign a quota of electricity to each household fairly and efficiently. Thus, \textsc{SLS} is a resource allocation problem. Any such a scheme must allocate all available supply (maximum efficiency) and maintain a steady revenue to the supplier \cite{walsh2015challenges}. A significant aspect of the resource allocation problem is to distribute electricity among households in such a way so that each household gets a fair amount of electricity \cite{lan2010axiomatic}. 
We have modeled the households' utilities by a parametric function ($\alpha$-fair) and formulated \textsc{SLS} problem as a social welfare maximization problem. Our parametric solution encompasses many well-known fairness notions as its special cases. We have also studied the effect of $\alpha$ on revenue generation considering block rate pricing. The parameter $\alpha$ can be tuned to find the optimal trade-off between fairness and generated revenue. We performed experiments on several datasets to extensively evaluate our method. We compare our results with several baselines and show that our model is effective in terms of increasing consumers satisfaction level as well as increasing the revenue for the utility companies at the same time. Note that throughout the paper, we are using the ``consumer" and ``household" word interchangeably.

The rest of the paper is organized as follows. We review some related work on demand-side management and fair allocation in Section \ref{relatedwork}. In Section \ref{Proposed_methodology}, we formulate the problem of \textsc{SLS}. We describe our approach towards solving the SLS problem in Section \ref{Solution}. In Section \ref{Experiments}, we describe experimental setup and dataset description while results and their comparisons are given in Section \ref{results_and_comparison}. The paper is finally concluded in Section \ref{Section_Conclusion}.

\section{Related Work}\label{relatedwork}
The problem of fair allocation of resources is deeply explored in different literatures related to computer sciences, communication systems, economics, game theory and social sciences \cite{jain1984quantitative,bredel2009understanding,bonald2001impact}. The mathematical foundation of $\alpha$ fair utility functions is formulated in \cite{lan2010axiomatic} in which a general form of fairness measure is presented, which is derived from five axioms (continuity, homogeneity, asymptotic saturation, irrelevance of partition, and monotonicity). 
A study for the allocation of rate and charging for a communication network has been carried out in \cite{kelly1997charging}. In this study, optimization problem is proposed supposing elastic traffic. The equilibrium and fairness criteria of problem are also explored. 
An online version of fair allocation of food for charity purpose is formulated in \cite{aleksandrov2015online}. The idea to use computation to increase both fairness and efficiency is presented in \cite{toby2015}. 
The problem of fair allocation, pareto optimality and efficiency are also well explained in \cite{chiang2012networked} with respect to communication networks. 

Although there is a vast amount of literature on topic related to resource allocation and fairness in different fields, but there is no solid work in this field related to fair and effective electricity allocation. The problem of soft load shedding is proposed in \cite{aslam2018soft}. However, the solution provided is based on qualitative reasoning, and no mathematical justification is provided. The notion of fairness is also not mentioned in the paper. Mansoor et al. in \cite{mansoor2020dr} studies a market model for Demand-Response (DR) using block rate pricing and propose a distributed algorithm to find the optimal pricing for each block
and the load. However, in their proposed method, both customers and the utility have to actively participate, which can make the system difficult to implement in real world scenario.
Chandan et al. in \cite{chandan2014idr} purposed a (DR) program controlled from utility that maximizes user convenience. However, this approach is possible with deep understanding of customers appliances usage pattern. Such kind of data is not available for all customers at utility scale. Our method, however, only require load data of consumers. In \cite{bashir2015delivering} a direct load control method that is capable to enforce several user defined low power states is presented. It directly controls the appliances of the house to manage the load. This method is only applicable if all appliances have the remote control capability. Secondly, implementing this technique on utility scale involves control of multi-million devices, which brings huge cost. Similarly, a queuing based energy management system for residential consumers is proposed by \cite{liu2016queuing}, which also require direct load control from the utility companies. 
However, as explained earlier it very costly to control these devices for million of customers. 

An efficient DR system is proposed in \cite{bashir2015aashiyana}, in which the authors propose a solution to match the demand with supply while trying to avoid complete blackout for majority of households and maximizing the satisfaction level of consumers.
An algorithm for fair load shedding scheme is proposed in \cite{oluwasuji2018algorithms}, in which authors divide the customers into groups so that the total electricity shortfall is equal to the total demand of each group. Then one group is selected to perform load shedding. However, their method performs complete blackout for a certain group, which can cause discomfort for the consumers.

Although a few \textsc{SLS} schemes have been proposed in the literature, however, no discussion regarding the fair allocation and revenue generation in the SLS schemes is discussed. We formulate the problem of \textsc{SLS} as a social welfare optimization problem. We introduced a parametric notion of fairness to the solution and studied the revenue generated under different allocations.

\section{Problem Formulation} \label{Proposed_methodology}
In this section we describe the requirements for SLS problem. Consider a set of $ N $ users where each user can be an individual customer or a household. 
These households are served by one power company. The total electricity supply available to the company is $S$. 
For every household $ i \in N $, 
there is a maximum power demand denoted by $ d_{i} $. Moreover, the sum of available demands is greater than the supply.
\begin{equation}
	\sum_{i=1}^{N} d_i \geq S
\end{equation}
To solve the problem of SLS, the power company has to assign each household an allocation of electricity $x_{i}$. Moreover, the allocation problem should clear the market i.e the aggregated allocations is equal to the supply. 
\begin{equation}
	0 \leq x_{i} \leq d_{i}, \ \ i \in N
\end{equation}
\begin{equation}
	\sum_{i=1}^{N} x_{i}=S
\end{equation}

In this paper, we have considered a general notion of fairness called $\alpha$-fair. An allocation $x_i^* $ is $\alpha$ fair if for any other allocation $x_i $, we obtained the following inequality from \cite{kelly1997charging} 
\begin{equation}\label{eq:16}
	\centering
	\sum_{i^=1}^{N}  \frac{x_{i} - x_i^{* \alpha}}{ x_i^{* \alpha}} \leq 0
\end{equation}  
where $\alpha \in [0,\infty]$. Different values of alpha produce different fairness measures with varying efficiency \cite{buzna2017controlling}. When $\alpha=0$ the efficiency (throughput) of the solution is maximum, and it favors larger allocations. When $\alpha=1$, the fairness measure becomes the most popular {\em proportional fair} \cite{altman2008generalized}. Proportional fair solution favors smaller allocations but less emphatically. As we increase $\alpha$ the fairness measure shifts favors from larger allocations to smaller one. An important characteristic of $\alpha$-fairness is that as we increase $\alpha$ the total throughput ($\sum_{i=1}^{N}x_i$) decreases. However, in SLS problem the efficiency is always maximum due to market clearance constraint ($\sum_{i=1}^{N} x_{i}=S$).

\section{SLS Problem}\label{Solution}
We take an optimization approach towards the SLS problem. To formulate the SLS optimization problem, we associate a utility function $U_i(x_{i})$ to $i^{th}$ consumer for $x_{i}$ consumption of electricity. The utility function measures the importance/satisfaction of the household as a function of the consumed electricity. We assume that the utility $U_i(x_{i})$ is an increasing, strictly concave and continuously differentiable, function of $x_{i}$ over $x_{i} \geq 0$ \cite{kelly1997charging}. Assume further that the utilities are additive , so that the total utility of all consumers is following:
\begin{equation}
	\sum\limits_{i=1}^{N} U_i(x_{i})
\end{equation}

To find optimal electricity quota for each consumer, the utility company has to solve the following social welfare optimization problem.
\begin{equation}\label{eq:4}
	\begin{array}{ll@{}ll}
		\text{maximize:}  & \displaystyle\sum\limits_{i=1}^{N} U_i(x_{i}) &\\ \\
		\text{subject to:}\\ 
		&\sum\limits_{i=1}^{N}x_{i} = S\\ \\
		& x_{i} \geq 0, & i \in N \\ \\
		&x_{i} \leq d_{i}, & i \in N \\ \\
	\end{array}
\end{equation}

The SLS problem from Equation \eqref{eq:4} is a convex optimization problem. The objective function is the sum of concave functions and is concave. The feasible region is an intersection of linear equality and inequalities, which is a convex set. First, we prove that solving the SLS problem is equivalent to maximizing the net utility of each household. Since the problem is convex, there is a unique solution, which can be found using the Lagrangian method.

\begin{equation}\label{eq:8}
	L(x,\lambda)=\sum\limits_{i=1}^{N} U_i(x_{i})+\lambda(S-\sum\limits_{i=1}^{N}x_{i})
\end{equation}

Where $\lambda$ is the Lagrange multiplier. By taking partial derivative with respect to $x_{i}$,
\begin{equation}\label{eq:9}
	\begin{aligned} 
		\frac{\partial L}{\partial x_i} & = U'(x_{i})-\lambda, \ \ i \in N 
	\end{aligned}
\end{equation}

Where $\lambda$ can be interpreted as per unit electricity price.
If every household is charged price $\lambda$, and is allowed to freely change the electricity units (demand), then each household wants to solve the following net utility maximization problem.
\begin{equation}\label{eq:5}
	\begin{array}{ll@{}ll}
		\text{maximize:}  & \displaystyle U_i(x_{i})-\lambda x_{i} &\\ \\
		\text{subject to:} \\ 
		& x_{i} \geq 0, & i \in N \\ \\
		&x_{i} \leq d_{i}, & i \in N \\ \\
	\end{array}
\end{equation} 
%
%

\begin{theorem}
	There exists a price variable $\lambda$, such that the solution vector of optimization problem in Equation \eqref{eq:5} also solves the optimization problem for Equation \eqref{eq:4}.
\end{theorem}
\begin{proof}
	
	Since the problem in Equation \eqref{eq:5} is convex, its solution can be found by taking derivative, which is given below.
	\begin{equation}\label{eq:10}
		U'(x_{i})=\lambda, \ \ i \in N \\ 
	\end{equation}
	This is same as equilibrium condition of Equation \eqref{eq:4} as shown in Equation \eqref{eq:9} 
\end{proof}
We define $U_{\alpha}(x_{i})$ as a general class of utility functions whose solution is $\alpha$-fair allocation of electricity. 
The general form of alpha fair utility function is given in Equation \eqref{eq:7}.
\begin{equation} \label{eq:7}
	U_\alpha(x_i)=\begin{cases}
		\dfrac {x_i^{1-\alpha}}{1-\alpha} & \alpha \geq 0,\alpha \neq 1 \\
		log(x_i) & \alpha=1 
	\end{cases}
\end{equation} 
If we put $\alpha=0$ in Equation \eqref{eq:7}, the objective function is equivalent to maximizing the throughput (efficiency) of the problem. Throughput fairness gives priority to larger allocations. If we use $\alpha=1$ the resulting allocation vector is called proportionally fair. Proportional fair electricity allocation is achieved by maximizing the sum of logarithms of received electricity $\sum_{i=1}^{N} log (x_{i})$. A system is called proportional fair if we cannot provide any customer with a larger fraction of electricity without reducing the proportion to those that are receiving a smaller fraction of electricity. Proportional fair gives priority to smaller allocations but less emphatically. Between $\alpha=0$ to $\alpha= \infty$ different fairness criterias originate. As we increase $\alpha$ the allocation priority moves from larger to smaller allocations. 

It is a well know fact that larger $\alpha$ means more fair solution. However, for a system with single supply,  proportional fair ($\alpha = 1$) is same as max-min fair ($\alpha=\infty$) \cite{kelly1997charging}. Generally, as we increase $\alpha$, the fairness of solution increases but the total efficiency/throughput ($\sum_{i=1}^{N}x_i$) decreases \cite{buzna2017controlling}. However, in the problem of SLS, the efficiency (throughput) always remain the same (maximum) due to the market clearance constraint (sum of allocation always equal to the supply) i.e. $\sum_{i=1}^{N}x_i=S$.

Recall that SLS is performed on the flexible loads that can be shifted to some other time (hour). In a scenario where there are not enough flexibility loads available with the majority of the consumers that can be disabled in the critical moment, the available supply will be equally divided among all consumers.

\section{Experimental Setup}\label{Experiments}
In this section, we first discuss dataset statistics for both real-world and synthetic datasets and give our performance metrics. Then we discuss the baseline methods, which we are using for the comparison with our proposed approach.

\subsection{Dataset Description and Performance Measure}\label{Dataset_Description_and_Performance_Measure}
For synthetic datasets, we generate data using Binomial and Uniform distributions with $100$ values (households) for each distribution. Increasing the number of households has no significant difference in terms of performance (not in terms of runtime), that is why we consider $100$ values only (so that results can be computed fast). 

For real world datasets, we use hourly consumption data from Australia \cite{lusis2017short}, Sweden \cite{javed2012forecasting}, and Ireland \cite{irish_dataset}.  
For each dataset, one day is randomly selected and load of $24$ hours are aggregated for that day. The statistics of all datasets after removing consumers with missing values are given in Table \ref{tbl:Dataset}. The aggregated load values of all households for randomly selected day of real world datasets are given in Figure \ref{fig_real_datasets_stats}.

\begin{table}[h!]
	\begin{center}
		\begin{tabular}{c c c c} 
			\hline
			Dataset & No. of Consumers & No. of Hours & Duration \\ [0.5ex] 
			\hline\hline
			Australia & 34 & 26304 &  1-July-2010 to 30-Jun-2013 \\
			\hline
			Sweden & 582 & 17544 & 1-Jan-2004 to 31-Dec-2005 \\ 
			\hline
			Ireland & 707 & 12865 & 14-Jul-2009 to 31-Dec-2010  \\
			[1ex] 
			\hline
		\end{tabular}
		\caption{Statistics of real world datasets}
		\label{tbl:Dataset}
	\end{center}
\end{table}

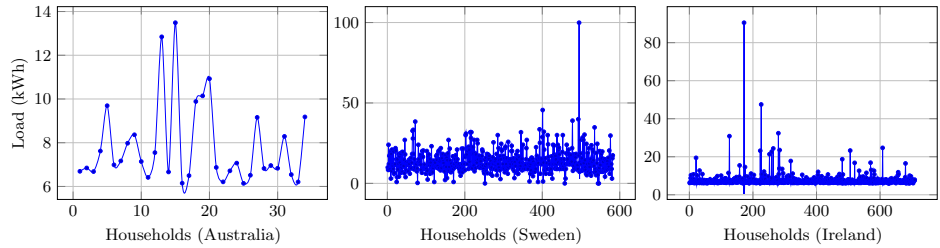
\begin{figure}[h!]
	\centering
	\footnotesize
	\begin{tikzpicture}[scale = 0.7]
	\begin{axis}[title={},
	compat=newest,
	xlabel style={text width=3.5cm, align=center},
	xlabel={{\small Households (Australia)}},
	ylabel={Load (kWh)}, 
	ylabel shift={-3pt},
	height=0.43\columnwidth, width=0.55\columnwidth, grid=major,
	]
	\addplot+[
	mark size=1pt,
	smooth,
	error bars/.cd,
	y fixed,
	y dir=both,
	y explicit
	] table [x={x}, y={y}, col sep=comma] {aus_datasets_stats.csv};
	\end{axis}
	\end{tikzpicture}%
	\begin{tikzpicture}[scale = 0.7]
	\begin{axis}[title={},
	compat=newest,
	xlabel style={text width=3.5cm, align=center},
	xlabel={{\small Households (Sweden)}},
	ylabel shift={-3pt},
	height=0.43\columnwidth, width=0.55\columnwidth, grid=major,
	]
	\addplot+[
	mark size=1pt,
	smooth,
	error bars/.cd,
	y fixed,
	y dir=both,
	y explicit
	] table [x={x}, y={y}, col sep=comma] {swe_datasets_stats.csv};
	\end{axis}
	\end{tikzpicture}%
	\begin{tikzpicture}[scale = 0.7]
	\begin{axis}[title={},
	compat=newest,
	xlabel style={text width=3.5cm, align=center},
	xlabel={{\small Households (Ireland)}},
	ylabel shift={-3pt},
	height=0.43\columnwidth, width=0.55\columnwidth, grid=major,
	]
	\addplot+[
	mark size=1pt,
	smooth,
	error bars/.cd,
	y fixed,
	y dir=both,
	y explicit
	] table [x={x}, y={y}, col sep=comma] {ire_datasets_stats.csv};
	\end{axis}
	\end{tikzpicture}%
	\caption{Aggregated Load of a randomly selected day ($24$ hours)}
	\label{fig_real_datasets_stats}
	
\end{figure}

For experiments, we tried different values of alpha $\alpha \in[0,10000]$. The available supply is varied from $60\%$ to $95\%$ of the total demand. All experiments are performed in Matlab using cvx toolbox on core $i3$ system with $4$GB memory.

We consider two metrics to measure the performance of our approach, $(i)$ Consumer Satisfaction Level (maximizing the utility) and $(ii)$ Revenue generation. Both metrics are described as follow:

\subsubsection{Consumer Satisfaction Level:}
To measure the satisfaction level of households, we use utility function described in \cite{bashir2015aashiyana}. We divide the satisfaction level of households into $5$ levels.
Depending on the percentage of allocation with respect to their demand, households will be placed in one of the $5$ levels.
The utility function is given in Equation \eqref{eq_utility_fun_ashiana}.

\begin{equation} \label{eq_utility_fun_ashiana}
	U(U_{max}, th_{U}, th_{L})=\begin{cases}
		U_{max} & for \ L5 \\	
		th_{U} & for \  L4 \\	
		\frac{th_{U} + th_{L}}{2} & for \  L3 \\	
		th_{L} & for \  L2 \\	
		0 & for \  L1 \\	
	\end{cases}
\end{equation} 
In Equation \eqref{eq_utility_fun_ashiana}, $U_{max}$ is the state where the household has highest satisfaction level ``$L5$" (i.e. household is getting allocation of electricity that is between $76\%$ to $100\%$ of its demand), while $L0$ means complete blackout (i.e. zero allocation). $L1$ to $L4$ are considered as restricted power states. The upper threshold $th_{U}$ and lower threshold $th_{L}$ contains $75\%$ and $25\%$ allocation respectively with respect to the demand of consumers. Our goal is to remove as much households from state $L1$ as possible and shift them to any of the upper states. In ideal scenario, all households should be placed in $L5$ provided that the shortfall constraint is satisfied.

\subsubsection{Revenue Generation:}
We consider revenue generation to measure the performance of our approach because utility companies want to maximize revenue and hence profit. Different utility companies use different pricing strategies like (flat tariff, block rate tariff, time of use tariff, real-time pricing etc.). In this study, we used block rate tariff because it is the most popular one (in most developing countries). Under block rate tariff, the revenue generated not only depends on total throughput but also the distribution of allocation vector.
In block rate tariff, different blocks of energy are charged at different rates. The rate per unit in each block is fixed. The succeeding block prices could be greater or lower than previous blocks. Since, in this study, we have considered a limited supply, increasing block rate pricing is used to discourage prodigious usage of electricity. A three block rate tariff with the increasing price is shown in Figure \ref{block_rate_tariff}.
\begin{figure}[H]
	\begin{center}
		\includegraphics[scale=0.5, page=1]{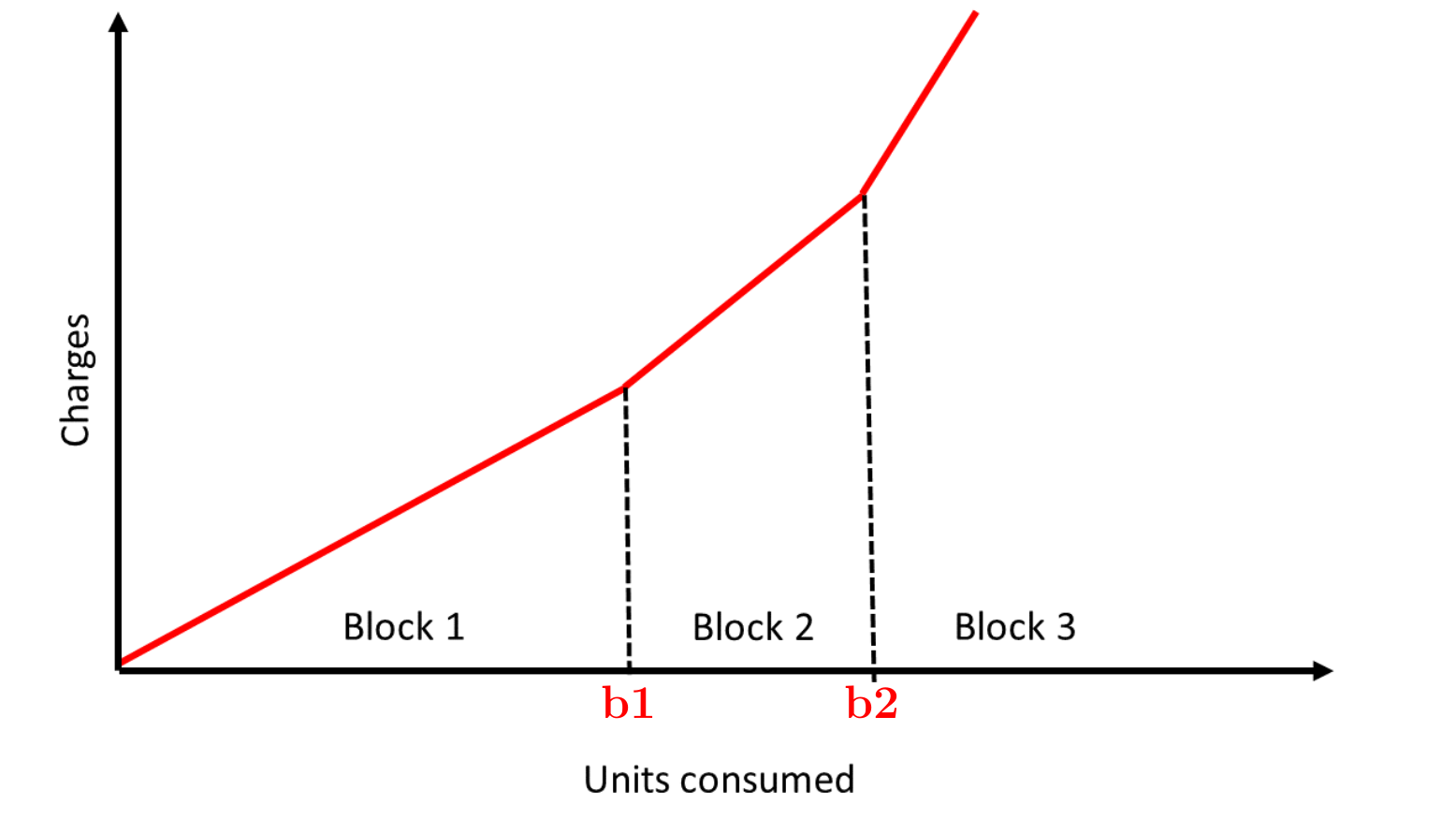}
		\caption{Block rate tariff. In our experiments, we are using two blocks with thresholds $b1$ and $b2$ and fix prices for the blocks that is $p_1=10$ and $p_2=20$}
		\label{block_rate_tariff}
	\end{center}
\end{figure}
The revenue generated from the $i^{th}$ household using block rate tariff with $T$ blocks is given below
\begin{equation} \label{eq:15}
	R_i=\begin{cases}
		p_1x_i & x_i \leq b_1 \\	
		p_1b_1+p_2(x_i-b_1) & b_1 < x_i \leq b_2 \\
		\vdots & \vdots \\
		p_1b_1+p_2(b_2-b_1)+...+p_T(x_i-b_T) & b_T < x_i 
	\end{cases}
\end{equation} 
Where $\{p_1,p_2,\ldots,p_T\}$ is the price associated with $\{block \ 1, block \ 2,\ldots,block \ T\}$. The total revenue generated is $\sum_{i=1}^{N} R_i$. Since $p_T > p_{T-1}... > p_1$, the fairness measure, which gives larger allocation priority will produce more revenue.
However, revenue is not the only objective we are concerned with as we have explained earlier that we also want the allocation to be fair. In this situation, there should be some compromise between revenue maximization and satisfaction of consumers in term of fairness.

For block rate pricing we are using two blocks with $p_1=10$ and $p_2=20$. Since, the revenue generated depends upon the selection of block threshold $b$, we tried different values of $b$ from $10$ to $90$ percentile of upper bound vector $d_i$, for $i \in N $.  

\subsection{Existing Methods Used for Comparison}\label{baseline_methods_section}
There are many simple baseline methods to execute the SLS mechanism. However, these methods are inefficient, or they result in a very unfair allocation. First, we considered a simple mechanism of equally dividing the available supply. This mechanism is called {\em equitable allocation} and defined as $x_i=\frac{S}{N},  i \in N$. However, equitable allocation is not an efficient solution because each household has a variable amount of demand. Allocating equal electricity to each household means that some households will get more electricity than their demand and vice versa.

Another baseline solution is percentage equitable. In this method, each household gets an equal fraction of electricity. For example, if there is an $80\%$ deficit of electricity, then each household will get $20\%$ of their demand. The drawback of this approach is that there are many households that are already using a minimal amount of electricity (only for their basic needs). Reducing their electricity demand will result in a very unfair allocation. Hence we are not using this baseline for comparison in our results.

The second method, which we use as a baseline, is {\em Max-Min fair allocation} \cite{le2005rate}. The primary objective of this algorithm is to maximize the minimum allocation for each household. This algorithm favors households that have smaller demand. The algorithm works as follows: We take the household with the minimum demand ($min(d_i)$, for $i \in N$) among the available set of household and give every household $min(d_i)$ amount of load. Repeat this process until $S=0$. When $S=0$, the quota assigned to each household will be considered as final allocation. Suppose if in a particular scenario, the available supply is not enough to assign the $min(d_i)$ to every household, then the remaining supply is equally divided among all households.

\section{Results and Comparison}\label{results_and_comparison}
In this section, we show the results in terms of consumer satisfaction and revenue generation for our proposed approach (see Equation \ref{eq:4} for our proposed approach) and compare them with the baseline (see Section \ref{baseline_methods_section} for baseline methods) using different datasets. 
\subsection{Results on Synthetic datasets}
Results for synthetic datasets are given in Figure \ref{fig_binomial_results}, which shows the number of consumers in each of the $5$ categories ($L5$, $L4$, $L3$, $L2$, $L1$) for dataset generated using binomial (above row) and uniform (bottom row) distribution for $100$ households and varying the shortfall percentage. We can see that as the shortfall increases, the overall satisfaction level of consumers decreases. However, in comparison with the baselines, our proposed approach satisfies a greater number of consumers. One important behavior to note here is that no approach place any consumer in the category $L1$ (complete blackout). This behavior shows that our approach (along with the baselines) is more effective as compared to the complete load shedding approaches.

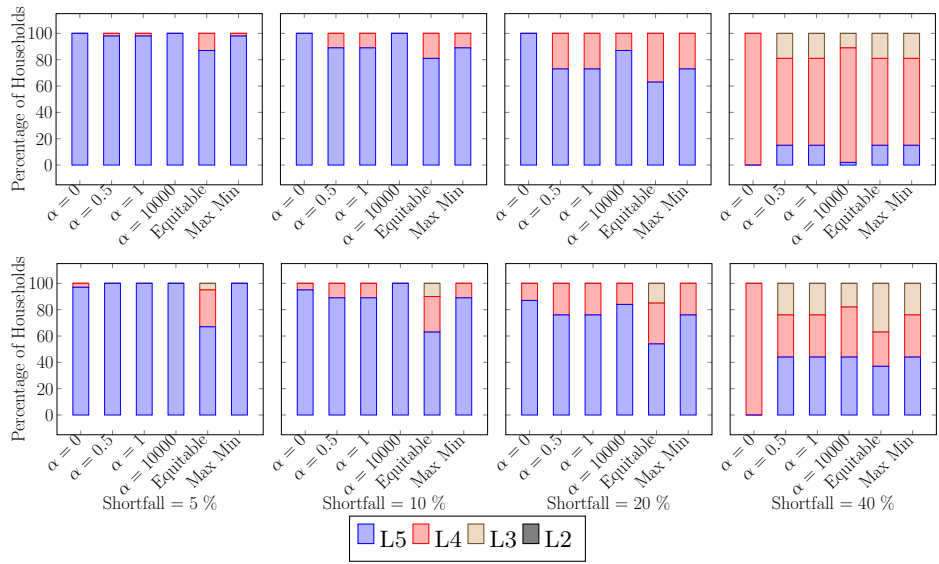
\begin{figure}[h!]
	\centering
	\begin{tikzpicture}[scale=0.40]
	\begin{axis}[
	ybar stacked,
	bar width=15pt,
	enlargelimits=0.15,
	legend columns=-1,
	legend entries={\strut L5, \strut L4, \strut L3, \strut L2},
	legend to name=CombinedLegendBar,
	ylabel={Percentage of Households},
	ylabel near ticks, 
	ylabel shift={-7pt},
	xlabel style={at={(0.5,-10ex)}},
	symbolic x coords={tool1, tool2, tool3, tool4, 
		tool5, tool6},
	xtick=data,
	xticklabels={$\alpha=0$,$\alpha=0.5$,$\alpha=1$,$\alpha=10000$,Equitable,Max Min},
	x tick label style={rotate=45,anchor=east},
	label style={font=\LARGE},
	tick label style={font=\LARGE},
	ymin=0,
	ymax=100,
	]
	\addplot+[ybar] plot coordinates {(tool1,100) (tool2,98) 
		(tool3,98) (tool4,100) (tool5,87) (tool6,98)};
	\addplot+[ybar] plot coordinates {(tool1,0) (tool2,2) 
		(tool3,2) (tool4,0) (tool5,13) (tool6,2)};
	\addplot+[ybar] plot coordinates {(tool1,0) (tool2,0)
		(tool3,0) (tool4,0) (tool5,0) (tool6,0)};
	\addplot+[ybar] plot coordinates {(tool1,0) (tool2,0)
		(tool3,0) (tool4,0) (tool5,0) (tool6,0)};
	\end{axis}
	\end{tikzpicture}%
	\begin{tikzpicture}[scale=0.40]
	\begin{axis}[
	ybar stacked,
	bar width=15pt,
	enlargelimits=0.15,
	legend columns=-1,
	legend entries={\strut L5, \strut L4, \strut L3, \strut L2},
	legend to name=CombinedLegendBar,
	xlabel style={at={(0.5,-10ex)}},
	symbolic x coords={tool1, tool2, tool3, tool4, 
		tool5, tool6},
	xtick=data,
	xticklabels={$\alpha=0$,$\alpha=0.5$,$\alpha=1$,$\alpha=10000$,Equitable,Max Min},
	x tick label style={rotate=45,anchor=east},
	label style={font=\LARGE},
	tick label style={font=\LARGE},
	yticklabels={},
	ymin=0,
	ymax=100,
	]
	\addplot+[ybar] plot coordinates {(tool1,100) (tool2,89) 
		(tool3,89) (tool4,100) (tool5,81) (tool6,89)};
	\addplot+[ybar] plot coordinates {(tool1,0) (tool2,11) 
		(tool3,11) (tool4,0) (tool5,19) (tool6,11)};
	\addplot+[ybar] plot coordinates {(tool1,0) (tool2,0)
		(tool3,0) (tool4,0) (tool5,0) (tool6,0)};
	\addplot+[ybar] plot coordinates {(tool1,0) (tool2,0)
		(tool3,0) (tool4,0) (tool5,0) (tool6,0)};
	\end{axis}
	\end{tikzpicture}%
	\begin{tikzpicture}[scale=0.40]
	\begin{axis}[
	ybar stacked,
	bar width=15pt,
	enlargelimits=0.15,
	legend columns=-1,
	legend entries={\strut L5, \strut L4, \strut L3, \strut L2},
	legend to name=CombinedLegendBar,
	xlabel style={at={(0.5,-10ex)}},
	symbolic x coords={tool1, tool2, tool3, tool4, 
		tool5, tool6},
	xtick=data,
	xticklabels={$\alpha=0$,$\alpha=0.5$,$\alpha=1$,$\alpha=10000$,Equitable,Max Min},
	x tick label style={rotate=45,anchor=east},
	label style={font=\LARGE},
	tick label style={font=\LARGE},
	yticklabels={},
	ymin=0,
	ymax=100,
	]
	\addplot+[ybar] plot coordinates {(tool1,100) (tool2,73) 
		(tool3,73) (tool4,87) (tool5,63) (tool6,73)};
	\addplot+[ybar] plot coordinates {(tool1,0) (tool2,27) 
		(tool3,27) (tool4,13) (tool5,37) (tool6,27)};
	\addplot+[ybar] plot coordinates {(tool1,0) (tool2,0)
		(tool3,0) (tool4,0) (tool5,0) (tool6,0)};
	\addplot+[ybar] plot coordinates {(tool1,0) (tool2,0)
		(tool3,0) (tool4,0) (tool5,0) (tool6,0)};
	\end{axis}
	\end{tikzpicture}%
	\begin{tikzpicture}[scale=0.40]
	\begin{axis}[
	ybar stacked,
	bar width=15pt,
	enlargelimits=0.15,
	legend columns=-1,
	legend entries={\strut L5, \strut L4, \strut L3, \strut L2},
	legend to name=CombinedLegendBar,
	xlabel style={at={(0.5,-10ex)}},
	symbolic x coords={tool1, tool2, tool3, tool4, 
		tool5, tool6},
	xtick=data,
	xticklabels={$\alpha=0$,$\alpha=0.5$,$\alpha=1$,$\alpha=10000$,Equitable,Max Min},
	x tick label style={rotate=45,anchor=east},
	label style={font=\LARGE},
	tick label style={font=\LARGE},
	yticklabels={},
	ymin=0,
	ymax=100,
	]
	\addplot+[ybar] plot coordinates {(tool1,0) (tool2,15) 
		(tool3,15) (tool4,2) (tool5,15) (tool6,15)};
	\addplot+[ybar] plot coordinates {(tool1,100) (tool2,66) 
		(tool3,66) (tool4,87) (tool5,66) (tool6,66)};
	\addplot+[ybar] plot coordinates {(tool1,0) (tool2,19)
		(tool3,19) (tool4,11) (tool5,19) (tool6,19)};
	\addplot+[ybar] plot coordinates {(tool1,0) (tool2,0)
		(tool3,0) (tool4,0) (tool5,0) (tool6,0)};
	\end{axis}
	\end{tikzpicture}%
	\\
	\begin{tikzpicture}[scale=0.40]
	\begin{axis}[
	ybar stacked,
	bar width=15pt,
	enlargelimits=0.15,
	legend columns=-1,
	legend entries={\strut L5, \strut L4, \strut L3, \strut L2},
	legend to name=CombinedLegendBar,
	ylabel={Percentage of Households},
	xlabel = {Shortfall = 5 \%},
	xlabel style={at={(0.5,-10ex)}},
	symbolic x coords={tool1, tool2, tool3, tool4, 
		tool5, tool6},
	xtick=data,
	xticklabels={$\alpha=0$,$\alpha=0.5$,$\alpha=1$,$\alpha=10000$,Equitable,Max Min},
	x tick label style={rotate=45,anchor=east},
	label style={font=\LARGE},
	tick label style={font=\LARGE},
	ymin=0,
	ymax=100,
	]
	\addplot+[ybar] plot coordinates {(tool1,97) (tool2,100) 
		(tool3,100) (tool4,100) (tool5,67) (tool6,100)};
	\addplot+[ybar] plot coordinates {(tool1,3) (tool2,0) 
		(tool3,0) (tool4,0) (tool5,28) (tool6,0)};
	\addplot+[ybar] plot coordinates {(tool1,0) (tool2,0)
		(tool3,0) (tool4,0) (tool5,5) (tool6,0)};
	\addplot+[ybar] plot coordinates {(tool1,0) (tool2,0)
		(tool3,0) (tool4,0) (tool5,0) (tool6,0)};
	\end{axis}
	\end{tikzpicture}%
	\begin{tikzpicture}[scale=0.40]
	\begin{axis}[
	ybar stacked,
	bar width=15pt,
	enlargelimits=0.15,
	legend columns=-1,
	legend entries={\strut L5, \strut L4, \strut L3, \strut L2},
	legend to name=CombinedLegendBar,
	xlabel = {Shortfall = 10 \%},
	xlabel style={at={(0.5,-10ex)}},
	symbolic x coords={tool1, tool2, tool3, tool4, 
		tool5, tool6},
	xtick=data,
	xticklabels={$\alpha=0$,$\alpha=0.5$,$\alpha=1$,$\alpha=10000$,Equitable,Max Min},
	x tick label style={rotate=45,anchor=east},
	label style={font=\LARGE},
	tick label style={font=\LARGE},
	yticklabels={},
	ymin=0,
	ymax=100,
	]
	\addplot+[ybar] plot coordinates {(tool1,95) (tool2,89) 
		(tool3,89) (tool4,100) (tool5,63) (tool6,89)};
	\addplot+[ybar] plot coordinates {(tool1,5) (tool2,11) 
		(tool3,11) (tool4,0) (tool5,27) (tool6,11)};
	\addplot+[ybar] plot coordinates {(tool1,0) (tool2,0)
		(tool3,0) (tool4,0) (tool5,10) (tool6,0)};
	\addplot+[ybar] plot coordinates {(tool1,0) (tool2,0)
		(tool3,0) (tool4,0) (tool5,0) (tool6,0)};
	\end{axis}
	\end{tikzpicture}%
	\begin{tikzpicture}[scale=0.40]
	\begin{axis}[
	ybar stacked,
	bar width=15pt,
	enlargelimits=0.15,
	legend columns=-1,
	legend entries={\strut L5, \strut L4, \strut L3, \strut L2},
	legend to name=CombinedLegendBar,
	xlabel = {Shortfall = 20 \%},
	xlabel style={at={(0.5,-10ex)}},
	symbolic x coords={tool1, tool2, tool3, tool4, 
		tool5, tool6},
	xtick=data,
	xticklabels={$\alpha=0$,$\alpha=0.5$,$\alpha=1$,$\alpha=10000$,Equitable,Max Min},
	x tick label style={rotate=45,anchor=east},
	label style={font=\LARGE},
	tick label style={font=\LARGE},
	yticklabels={},
	ymin=0,
	ymax=100,
	]
	\addplot+[ybar] plot coordinates {(tool1,87) (tool2,76) 
		(tool3,76) (tool4,84) (tool5,54) (tool6,76)};
	\addplot+[ybar] plot coordinates {(tool1,13) (tool2,24) 
		(tool3,24) (tool4,16) (tool5,31) (tool6,24)};
	\addplot+[ybar] plot coordinates {(tool1,0) (tool2,0)
		(tool3,0) (tool4,0) (tool5,15) (tool6,0)};
	\addplot+[ybar] plot coordinates {(tool1,0) (tool2,0)
		(tool3,0) (tool4,0) (tool5,0) (tool6,0)};
	\end{axis}
	\end{tikzpicture}%
	\begin{tikzpicture}[scale=0.40]
	\begin{axis}[
	ybar stacked,
	bar width=15pt,
	enlargelimits=0.15,
	legend columns=-1,
	legend entries={\strut L5, \strut L4, \strut L3, \strut L2},
	legend to name=CombinedLegendBar,
	xlabel = {Shortfall = 40 \%},
	xlabel style={at={(0.5,-10ex)}},
	symbolic x coords={tool1, tool2, tool3, tool4, 
		tool5, tool6},
	xtick=data,
	xticklabels={$\alpha=0$,$\alpha=0.5$,$\alpha=1$,$\alpha=10000$,Equitable,Max Min},
	x tick label style={rotate=45,anchor=east},
	label style={font=\LARGE},
	tick label style={font=\LARGE},
	yticklabels={},
	ymin=0,
	ymax=100,
	]
	\addplot+[ybar] plot coordinates {(tool1,0) (tool2,44) 
		(tool3,44) (tool4,44) (tool5,37) (tool6,44)};
	\addplot+[ybar] plot coordinates {(tool1,100) (tool2,32) 
		(tool3,32) (tool4,38) (tool5,26) (tool6,32)};
	\addplot+[ybar] plot coordinates {(tool1,0) (tool2,24)
		(tool3,24) (tool4,18) (tool5,37) (tool6,24)};
	\addplot+[ybar] plot coordinates {(tool1,0) (tool2,0)
		(tool3,0) (tool4,0) (tool5,0) (tool6,0)};
	\end{axis}
	\end{tikzpicture}
	\ref{CombinedLegendBar}
	\caption{Comparison (using stacked bar plot) of our proposed approach using different values of $\alpha$ (see Equation \ref{eq:4} for our proposed approach) with max-min fair and equitable allocation using data generated from binomial (top row) and uniform (bottom row) distribution for $100$ households. Figure is best seen in color}
	\label{fig_binomial_results}
\end{figure}

\subsection{Runtime Analysis}
With the increasing number of households, the runtime of our proposed approach remains almost linear. Figure \ref{fig_convex_runtime} shows the time (in seconds) of our approach with increasing number of households using dataset generated from binomial distribution (see Section \ref{Dataset_Description_and_Performance_Measure} for detail regarding datasets).

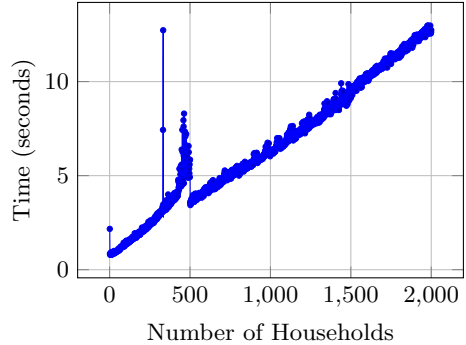
\begin{figure}[h!]
	\centering
	\footnotesize
	\begin{tikzpicture}
	\begin{axis}[title={},
	compat=newest,
	xlabel style={text width=3.5cm, align=center},
	xlabel={{\small Number of Households}},
	ylabel={Time (seconds)}, 
	ylabel shift={-3pt},
	height=0.43\columnwidth, width=0.55\columnwidth, grid=major,
	]
	\addplot+[
	mark size=1pt,
	smooth,
	error bars/.cd,
	y fixed,
	y dir=both,
	y explicit
	] table [x={x}, y={y}, col sep=comma] {runtime_from_customers_convex_final.csv};
	\end{axis}
	\end{tikzpicture}%
	\caption{Runtime in seconds for data generated from binomial distribution with alpha = $2$ and shortfall of $20 \%$}
	\label{fig_convex_runtime}
\end{figure}

\subsection{Results on Real-world datasets}
The results for real-world datasets, namely Sweden, Australia, and Ireland are given in Figure \ref{fig_sweden_results}. We observed similar behavior for real-world datasets as compared to the synthetic dataset.

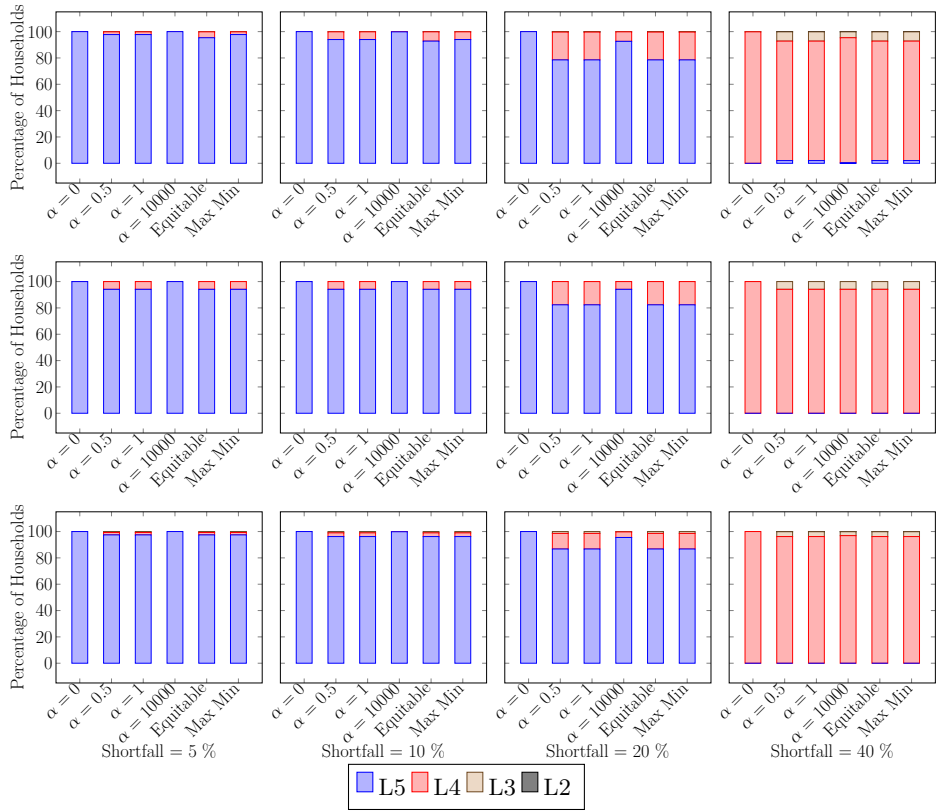
\begin{figure}[h!]
	\centering
	\begin{tikzpicture}[scale=0.40]
	\begin{axis}[
	ybar stacked,
	bar width=15pt,
	enlargelimits=0.15,
	legend columns=-1,
	legend entries={\strut L5, \strut L4, \strut L3, \strut L2},
	legend to name=CombinedLegendBar,
	ylabel={Percentage of Households},
	xlabel style={at={(0.5,-10ex)}},
	symbolic x coords={tool1, tool2, tool3, tool4, 
		tool5, tool6},
	xtick=data,
	xticklabels={$\alpha=0$,$\alpha=0.5$,$\alpha=1$,$\alpha=10000$,Equitable,Max Min},
	x tick label style={rotate=45,anchor=east},
	label style={font=\LARGE},
	tick label style={font=\LARGE},
	ymin=0,
	ymax=100,
	]
	\addplot+[ybar] plot coordinates {(tool1,100) (tool2,97.766) 
		(tool3,97.766) (tool4,100) (tool5,95.361) (tool6,97.766)};
	\addplot+[ybar] plot coordinates {(tool1,0) (tool2,2.0619) 
		(tool3,2.0619) (tool4,0) (tool5,4.4674) (tool6,2.0619)};
	\addplot+[ybar] plot coordinates {(tool1,0) (tool2,0.17182)
		(tool3,0.17182) (tool4,0) (tool5,0.17182) (tool6,0.17182)};
	\addplot+[ybar] plot coordinates {(tool1,0) (tool2,0)
		(tool3,0) (tool4,0) (tool5,0) (tool6,0)};
	\end{axis}
	\end{tikzpicture}%
	\begin{tikzpicture}[scale=0.40]
	\begin{axis}[
	ybar stacked,
	bar width=15pt,
	enlargelimits=0.15,
	legend columns=-1,
	legend entries={\strut L5, \strut L4, \strut L3, \strut L2},
	legend to name=CombinedLegendBar,
	xlabel style={at={(0.5,-10ex)}},
	symbolic x coords={tool1, tool2, tool3, tool4, 
		tool5, tool6},
	xtick=data,
	xticklabels={$\alpha=0$,$\alpha=0.5$,$\alpha=1$,$\alpha=10000$,Equitable,Max Min},
	x tick label style={rotate=45,anchor=east},
	label style={font=\LARGE},
	tick label style={font=\LARGE},
	yticklabels={},
	ymin=0,
	ymax=100,
	]
	\addplot+[ybar] plot coordinates {(tool1,100) (tool2,93.986) 
		(tool3,93.986) (tool4,99.828) (tool5,92.784) (tool6,93.986)};
	\addplot+[ybar] plot coordinates {(tool1,0) (tool2,5.8419) 
		(tool3,5.8419) (tool4,0.17182) (tool5,7.0447) (tool6,5.8419)};
	\addplot+[ybar] plot coordinates {(tool1,0) (tool2,0.17182)
		(tool3,0.17182) (tool4,0) (tool5,0) (tool6,0.17182)};
	\addplot+[ybar] plot coordinates {(tool1,0) (tool2,0)
		(tool3,0) (tool4,0) (tool5,0.17182) (tool6,0)};
	\end{axis}
	\end{tikzpicture}%
	\begin{tikzpicture}[scale=0.40]
	\begin{axis}[
	ybar stacked,
	bar width=15pt,
	enlargelimits=0.15,
	legend columns=-1,
	legend entries={\strut L5, \strut L4, \strut L3, \strut L2},
	legend to name=CombinedLegendBar,
	xlabel style={at={(0.5,-10ex)}},
	symbolic x coords={tool1, tool2, tool3, tool4, 
		tool5, tool6},
	xtick=data,
	xticklabels={$\alpha=0$,$\alpha=0.5$,$\alpha=1$,$\alpha=10000$,Equitable,Max Min},
	x tick label style={rotate=45,anchor=east},
	label style={font=\LARGE},
	tick label style={font=\LARGE},
	yticklabels={},
	ymin=0,
	ymax=100,
	]
	\addplot+[ybar] plot coordinates {(tool1,100) (tool2,78.522) 
		(tool3,78.522) (tool4,92.612) (tool5,78.522) (tool6,78.522)};
	\addplot+[ybar] plot coordinates {(tool1,0) (tool2,21.134) 
		(tool3,21.134) (tool4,7.2165) (tool5,21.134) (tool6,21.134)};
	\addplot+[ybar] plot coordinates {(tool1,0) (tool2,0.17182)
		(tool3,0.17182) (tool4,0.17182) (tool5,0.17182) (tool6,0.17182)};
	\addplot+[ybar] plot coordinates {(tool1,0) (tool2,0.17182)
		(tool3,0.17182) (tool4,0) (tool5,0.17182) (tool6,0.17182)};
	\end{axis}
	\end{tikzpicture}%
	\begin{tikzpicture}[scale=0.40]
	\begin{axis}[
	ybar stacked,
	bar width=15pt,
	enlargelimits=0.15,
	legend columns=-1,
	legend entries={\strut L5, \strut L4, \strut L3, \strut L2},
	legend to name=CombinedLegendBar,
	xlabel style={at={(0.5,-10ex)}},
	symbolic x coords={tool1, tool2, tool3, tool4, 
		tool5, tool6},
	xtick=data,
	xticklabels={$\alpha=0$,$\alpha=0.5$,$\alpha=1$,$\alpha=10000$,Equitable,Max Min},
	x tick label style={rotate=45,anchor=east},
	label style={font=\LARGE},
	tick label style={font=\LARGE},
	yticklabels={},
	ymin=0,
	ymax=100,
	]
	\addplot+[ybar] plot coordinates {(tool1,0) (tool2,2.0619) 
		(tool3,2.0619) (tool4,0.34364) (tool5,2.0619) (tool6,2.0619)};
	\addplot+[ybar] plot coordinates {(tool1,99.828) (tool2,90.722) 
		(tool3,90.722) (tool4,95.017) (tool5,90.722) (tool6,90.722)};
	\addplot+[ybar] plot coordinates {(tool1,0.17182) (tool2,7.0447)
		(tool3,7.0447) (tool4,4.4674) (tool5,7.0447) (tool6,7.0447)};
	\addplot+[ybar] plot coordinates {(tool1,0) (tool2,0.17182)
		(tool3,0.17182) (tool4,0.17182) (tool5,0.17182) (tool6,0.17182)};
	\end{axis}
	\end{tikzpicture}%
	\\
	\begin{tikzpicture}[scale=0.40]
	\begin{axis}[
	ybar stacked,
	bar width=15pt,
	enlargelimits=0.15,
	legend columns=-1,
	legend entries={\strut L5, \strut L4, \strut L3, \strut L2},
	legend to name=CombinedLegendBar,
	ylabel={Percentage of Households},
	xlabel style={at={(0.5,-10ex)}},
	symbolic x coords={tool1, tool2, tool3, tool4, 
		tool5, tool6},
	xtick=data,
	xticklabels={$\alpha=0$,$\alpha=0.5$,$\alpha=1$,$\alpha=10000$,Equitable,Max Min},
	x tick label style={rotate=45,anchor=east},
	label style={font=\LARGE},
	tick label style={font=\LARGE},
	ymin=0,
	ymax=100,
	]
	\addplot+[ybar] plot coordinates {(tool1,100) (tool2,94.118) 
		(tool3,94.118) (tool4,100) (tool5,94.118) (tool6,94.118)};
	\addplot+[ybar] plot coordinates {(tool1,0) (tool2,5.8824) 
		(tool3,5.8824) (tool4,0) (tool5,5.8824) (tool6,5.8824)};
	\addplot+[ybar] plot coordinates {(tool1,0) (tool2,0)
		(tool3,0) (tool4,0) (tool5,0) (tool6,0)};
	\addplot+[ybar] plot coordinates {(tool1,0) (tool2,0)
		(tool3,0) (tool4,0) (tool5,0) (tool6,0)};
	\end{axis}
	\end{tikzpicture}%
	\begin{tikzpicture}[scale=0.40]
	\begin{axis}[
	ybar stacked,
	bar width=15pt,
	enlargelimits=0.15,
	legend columns=-1,
	legend entries={\strut L5, \strut L4, \strut L3, \strut L2},
	legend to name=CombinedLegendBar,
	xlabel style={at={(0.5,-10ex)}},
	symbolic x coords={tool1, tool2, tool3, tool4, 
		tool5, tool6},
	xtick=data,
	xticklabels={$\alpha=0$,$\alpha=0.5$,$\alpha=1$,$\alpha=10000$,Equitable,Max Min},
	x tick label style={rotate=45,anchor=east},
	label style={font=\LARGE},
	tick label style={font=\LARGE},
	yticklabels={},
	ymin=0,
	ymax=100,
	]
	\addplot+[ybar] plot coordinates {(tool1,100) (tool2,94.118) 
		(tool3,94.118) (tool4,100) (tool5,94.118) (tool6,94.118)};
	\addplot+[ybar] plot coordinates {(tool1,0) (tool2,5.8824) 
		(tool3,5.8824) (tool4,0) (tool5,5.8824) (tool6,5.8824)};
	\addplot+[ybar] plot coordinates {(tool1,0) (tool2,0)
		(tool3,0) (tool4,0) (tool5,0) (tool6,0)};
	\addplot+[ybar] plot coordinates {(tool1,0) (tool2,0)
		(tool3,0) (tool4,0) (tool5,0) (tool6,0)};
	\end{axis}
	\end{tikzpicture}%
	\begin{tikzpicture}[scale=0.40]
	\begin{axis}[
	ybar stacked,
	bar width=15pt,
	enlargelimits=0.15,
	legend columns=-1,
	legend entries={\strut L5, \strut L4, \strut L3, \strut L2},
	legend to name=CombinedLegendBar,
	xlabel style={at={(0.5,-10ex)}},
	symbolic x coords={tool1, tool2, tool3, tool4, 
		tool5, tool6},
	xtick=data,
	xticklabels={$\alpha=0$,$\alpha=0.5$,$\alpha=1$,$\alpha=10000$,Equitable,Max Min},
	x tick label style={rotate=45,anchor=east},
	label style={font=\LARGE},
	tick label style={font=\LARGE},
	yticklabels={},
	ymin=0,
	ymax=100,
	]
	\addplot+[ybar] plot coordinates {(tool1,100) (tool2,82.353) 
		(tool3,82.353) (tool4,94.118) (tool5,82.353) (tool6,82.353)};
	\addplot+[ybar] plot coordinates {(tool1,0) (tool2,17.647) 
		(tool3,17.647) (tool4,5.8824) (tool5,17.647) (tool6,17.647)};
	\addplot+[ybar] plot coordinates {(tool1,0) (tool2,0)
		(tool3,0) (tool4,0) (tool5,0) (tool6,0)};
	\addplot+[ybar] plot coordinates {(tool1,0) (tool2,0)
		(tool3,0) (tool4,0) (tool5,0) (tool6,0)};
	\end{axis}
	\end{tikzpicture}%
	\begin{tikzpicture}[scale=0.40]
	\begin{axis}[
	ybar stacked,
	bar width=15pt,
	enlargelimits=0.15,
	legend columns=-1,
	legend entries={\strut L5, \strut L4, \strut L3, \strut L2},
	legend to name=CombinedLegendBar,
	xlabel style={at={(0.5,-10ex)}},
	symbolic x coords={tool1, tool2, tool3, tool4, 
		tool5, tool6},
	xtick=data,
	xticklabels={$\alpha=0$,$\alpha=0.5$,$\alpha=1$,$\alpha=10000$,Equitable,Max Min},
	x tick label style={rotate=45,anchor=east},
	label style={font=\LARGE},
	tick label style={font=\LARGE},
	yticklabels={},
	ymin=0,
	ymax=100,
	]
	\addplot+[ybar] plot coordinates {(tool1,0) (tool2,0) 
		(tool3,0) (tool4,0) (tool5,0) (tool6,0)};
	\addplot+[ybar] plot coordinates {(tool1,100) (tool2,94.118) 
		(tool3,94.118) (tool4,94.118) (tool5,94.118) (tool6,94.118)};
	\addplot+[ybar] plot coordinates {(tool1,0) (tool2,5.8824)
		(tool3,5.8824) (tool4,5.8824) (tool5,5.8824) (tool6,5.8824)};
	\addplot+[ybar] plot coordinates {(tool1,0) (tool2,0)
		(tool3,0) (tool4,0) (tool5,0) (tool6,0)};
	\end{axis}
	\end{tikzpicture}%
	\\
	\begin{tikzpicture}[scale=0.40]
	\begin{axis}[
	ybar stacked,
	bar width=15pt,
	enlargelimits=0.15,
	legend columns=-1,
	legend entries={\strut L5, \strut L4, \strut L3, \strut L2},
	legend to name=CombinedLegendBar,
	ylabel={Percentage of Households},
	xlabel = {Shortfall = 5 \%},
	xlabel style={at={(0.5,-10ex)}},
	symbolic x coords={tool1, tool2, tool3, tool4, 
		tool5, tool6},
	xtick=data,
	xticklabels={$\alpha=0$,$\alpha=0.5$,$\alpha=1$,$\alpha=10000$,Equitable,Max Min},
	x tick label style={rotate=45,anchor=east},
	label style={font=\LARGE},
	tick label style={font=\LARGE},
	ymin=0,
	ymax=100,
	]
	\addplot+[ybar] plot coordinates {(tool1,100) (tool2,97.461) 
		(tool3,97.461) (tool4,100) (tool5,97.461) (tool6,97.461)};
	\addplot+[ybar] plot coordinates {(tool1,0) (tool2,1.8336) 
		(tool3,1.8336) (tool4,0) (tool5,1.8336) (tool6,1.8336)};
	\addplot+[ybar] plot coordinates {(tool1,0) (tool2,0.70522)
		(tool3,0.70522) (tool4,0) (tool5,0.70522) (tool6,0.70522)};
	\addplot+[ybar] plot coordinates {(tool1,0) (tool2,0)
		(tool3,0) (tool4,0) (tool5,0) (tool6,0)};
	\end{axis}
	\end{tikzpicture}%
	\begin{tikzpicture}[scale=0.40]
	\begin{axis}[
	ybar stacked,
	bar width=15pt,
	enlargelimits=0.15,
	legend columns=-1,
	legend entries={\strut L5, \strut L4, \strut L3, \strut L2},
	legend to name=CombinedLegendBar,
	xlabel = {Shortfall = 10 \%},
	xlabel style={at={(0.5,-10ex)}},
	symbolic x coords={tool1, tool2, tool3, tool4, 
		tool5, tool6},
	xtick=data,
	xticklabels={$\alpha=0$,$\alpha=0.5$,$\alpha=1$,$\alpha=10000$,Equitable,Max Min},
	x tick label style={rotate=45,anchor=east},
	label style={font=\LARGE},
	tick label style={font=\LARGE},
	yticklabels={},
	ymin=0,
	ymax=100,
	]
	\addplot+[ybar] plot coordinates {(tool1,100) (tool2,96.192) 
		(tool3,96.192) (tool4,99.859) (tool5,96.192) (tool6,96.192)};
	\addplot+[ybar] plot coordinates {(tool1,0) (tool2,2.8209) 
		(tool3,2.8209) (tool4,0.14104) (tool5,2.8209) (tool6,2.8209)};
	\addplot+[ybar] plot coordinates {(tool1,0) (tool2,0.98731)
		(tool3,0.98731) (tool4,0) (tool5,0.98731) (tool6,0.98731)};
	\addplot+[ybar] plot coordinates {(tool1,0) (tool2,0)
		(tool3,0) (tool4,0) (tool5,0) (tool6,0)};
	\end{axis}
	\end{tikzpicture}%
	\begin{tikzpicture}[scale=0.40]
	\begin{axis}[
	ybar stacked,
	bar width=15pt,
	enlargelimits=0.15,
	legend columns=-1,
	legend entries={\strut L5, \strut L4, \strut L3, \strut L2},
	legend to name=CombinedLegendBar,
	xlabel = {Shortfall = 20 \%},
	xlabel style={at={(0.5,-10ex)}},
	symbolic x coords={tool1, tool2, tool3, tool4, 
		tool5, tool6},
	xtick=data,
	xticklabels={$\alpha=0$,$\alpha=0.5$,$\alpha=1$,$\alpha=10000$,Equitable,Max Min},
	x tick label style={rotate=45,anchor=east},
	label style={font=\LARGE},
	tick label style={font=\LARGE},
	yticklabels={},
	ymin=0,
	ymax=100,
	]
	\addplot+[ybar] plot coordinates {(tool1,100) (tool2,86.742) 
		(tool3,86.742) (tool4,95.487) (tool5,86.742) (tool6,86.742)};
	\addplot+[ybar] plot coordinates {(tool1,0) (tool2,11.848) 
		(tool3,11.848) (tool4,4.2313) (tool5,11.848) (tool6,11.848)};
	\addplot+[ybar] plot coordinates {(tool1,0) (tool2,1.4104)
		(tool3,1.4104) (tool4,0.28209) (tool5,1.4104) (tool6,1.4104)};
	\addplot+[ybar] plot coordinates {(tool1,0) (tool2,0)
		(tool3,0) (tool4,0) (tool5,0) (tool6,0)};
	\end{axis}
	\end{tikzpicture}%
	\begin{tikzpicture}[scale=0.40]
	\begin{axis}[
	ybar stacked,
	bar width=15pt,
	enlargelimits=0.15,
	legend columns=-1,
	legend entries={\strut L5, \strut L4, \strut L3, \strut L2},
	legend to name=CombinedLegendBar,
	xlabel = {Shortfall = 40 \%},
	xlabel style={at={(0.5,-10ex)}},
	symbolic x coords={tool1, tool2, tool3, tool4, 
		tool5, tool6},
	xtick=data,
	xticklabels={$\alpha=0$,$\alpha=0.5$,$\alpha=1$,$\alpha=10000$,Equitable,Max Min},
	x tick label style={rotate=45,anchor=east},
	label style={font=\LARGE},
	tick label style={font=\LARGE},
	yticklabels={},
	ymin=0,
	ymax=100,
	]
	\addplot+[ybar] plot coordinates {(tool1,0) (tool2,0) 
		(tool3,0) (tool4,0) (tool5,0) (tool6,0)};
	\addplot+[ybar] plot coordinates {(tool1,100) (tool2,96.192)
		(tool3,96.192) (tool4,96.897) (tool5,96.192) (tool6,96.192)};
	\addplot+[ybar] plot coordinates {(tool1,0) (tool2,3.6671) 
		(tool3,3.6671) (tool4,2.9619) (tool5,3.6671) (tool6,3.6671)};
	\addplot+[ybar] plot coordinates {(tool1,0) (tool2,0.14104)
		(tool3,0.14104) (tool4,0.14104) (tool5,0.14104) (tool6,0.14104)};
	\end{axis}
	\end{tikzpicture}
	\ref{CombinedLegendBar}
	\caption{Comparison of our proposed approach with max-min fair and equitable allocation on Sweden (top row), Australia (middle row), and Ireland (bottom row) dataset. Figure is best seen in color}
	\label{fig_sweden_results}
\end{figure}

To evaluate the behavior of our method on the different types of consumers, we categorize the consumers into low demand, medium demand, and high demand consumers. To categorize the consumers, a consumer with highest load (of a randomly selected day) is taken from each dataset and its load is divided by $3$ to get the thresholds for low, medium, and high category. For each category, we separately computed results (using the load of same randomly selected day as described above) to see the behavior of our algorithm. Results for different categories of consumers for Sweden dataset with the shortfall of $20 \%$ and $40\%$ are shown in Figure \ref{fig_sweden_cust_category}. We can observe that our algorithm is more inclined towards the low and medium category of consumers. As we move towards high demand consumers, the algorithm starts to place them in lower categories (i.e. $L4$). This behavior is due to the fact that since the demand for high-end consumers is very large, satisfying their needs while fulfilling the available supply constraint ($\sum_{i=1}^{N} x_{i}=S$) is very difficult.
We conclude that our algorithm favors low demand consumers more. Since the number of low demand consumers is greater than the high demand consumers, the algorithm tries to satisfy the majority (low demand consumers) and hence fulfilling the fairness criteria.
Similar behavior is observed for other datasets as well. Note that in this paper, we have studied the allocation problem for residential consumers only. Allocation for other sectors, such as industrial and commercial areas is beyond the scope of this study.

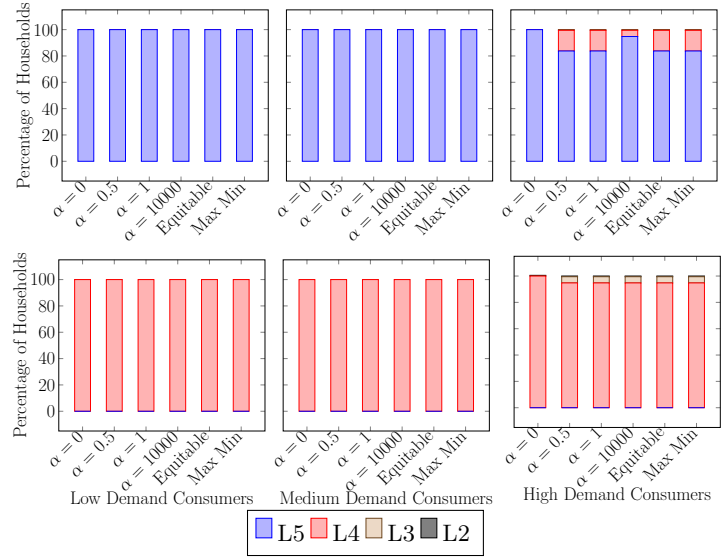
\begin{figure}[h!]
	\centering
	\begin{tikzpicture}[scale=0.40]
	\begin{axis}[
	ybar stacked,
	bar width=15pt,
	enlargelimits=0.15,
	legend columns=-1,
	legend entries={\strut L5, \strut L4, \strut L3, \strut L2},
	legend to name=CombinedLegendBar,
	ylabel={Percentage of Households},
	xlabel style={at={(0.5,-10ex)}},
	symbolic x coords={tool1, tool2, tool3, tool4, 
		tool5, tool6},
	xtick=data,
	xticklabels={$\alpha=0$,$\alpha=0.5$,$\alpha=1$,$\alpha=10000$,Equitable,Max Min},
	x tick label style={rotate=45,anchor=east},
	label style={font=\LARGE},
	tick label style={font=\LARGE},
	ymin=0,
	ymax=100,
	]
	\addplot+[ybar] plot coordinates {(tool1,100) (tool2,100) 
		(tool3,100) (tool4,100) (tool5,100) (tool6,100)};
	\addplot+[ybar] plot coordinates {(tool1,0) (tool2,0) 
		(tool3,0) (tool4,0) (tool5,0) (tool6,0)};
	\addplot+[ybar] plot coordinates {(tool1,0) (tool2,0) 
		(tool3,0) (tool4,0) (tool5,0) (tool6,0)};
	\addplot+[ybar] plot coordinates {(tool1,0) (tool2,0) 
		(tool3,0) (tool4,0) (tool5,0) (tool6,0)};
	\end{axis}
	\end{tikzpicture}%
	\begin{tikzpicture}[scale=0.40]
	\begin{axis}[
	ybar stacked,
	bar width=15pt,
	enlargelimits=0.15,
	legend columns=-1,
	legend entries={\strut L5, \strut L4, \strut L3, \strut L2},
	legend to name=CombinedLegendBar,
	xlabel style={at={(0.5,-10ex)}},
	symbolic x coords={tool1, tool2, tool3, tool4, 
		tool5, tool6},
	xtick=data,
	xticklabels={$\alpha=0$,$\alpha=0.5$,$\alpha=1$,$\alpha=10000$,Equitable,Max Min},
	x tick label style={rotate=45,anchor=east},
	label style={font=\LARGE},
	tick label style={font=\LARGE},
	yticklabels={},
	ymin=0,
	ymax=100,
	]
	\addplot+[ybar] plot coordinates {(tool1,100) (tool2,100) 
		(tool3,100) (tool4,100) (tool5,100) (tool6,100)};
	\addplot+[ybar] plot coordinates {(tool1,0) (tool2,0) 
		(tool3,0) (tool4,0) (tool5,0) (tool6,0)};
	\addplot+[ybar] plot coordinates {(tool1,0) (tool2,0) 
		(tool3,0) (tool4,0) (tool5,0) (tool6,0)};
	\addplot+[ybar] plot coordinates {(tool1,0) (tool2,0) 
		(tool3,0) (tool4,0) (tool5,0) (tool6,0)};
	\end{axis}
	\end{tikzpicture}%
	\begin{tikzpicture}[scale=0.40]
	\begin{axis}[
	ybar stacked,
	bar width=15pt,
	enlargelimits=0.15,
	legend columns=-1,
	legend entries={\strut L5, \strut L4, \strut L3, \strut L2},
	legend to name=CombinedLegendBar,
	xlabel style={at={(0.5,-10ex)}},
	symbolic x coords={tool1, tool2, tool3, tool4, 
		tool5, tool6},
	xtick=data,
	xticklabels={$\alpha=0$,$\alpha=0.5$,$\alpha=1$,$\alpha=10000$,Equitable,Max Min},
	x tick label style={rotate=45,anchor=east},
	label style={font=\LARGE},
	tick label style={font=\LARGE},
	yticklabels={},
	ymin=0,
	ymax=100,
	]
	\addplot+[ybar] plot coordinates {(tool1,100) (tool2,83.77) 
		(tool3,83.77) (tool4,94.764) (tool5,83.77) (tool6,83.77)};
	\addplot+[ybar] plot coordinates {(tool1,0) (tool2,15.707) 
		(tool3,15.707) (tool4,4.712) (tool5,15.707) (tool6,15.707)};
	\addplot+[ybar] plot coordinates {(tool1,0) (tool2,0.52356)
		(tool3,0.52356) (tool4,0.52356) (tool5,0.52356) (tool6,0.52356)};
	\addplot+[ybar] plot coordinates {(tool1,0) (tool2,0)
		(tool3,0) (tool4,0) (tool5,0) (tool6,0)};
	\end{axis}
	\end{tikzpicture}%
	\\
	\begin{tikzpicture}[scale=0.40]
	\begin{axis}[
	ybar stacked,
	bar width=15pt,
	enlargelimits=0.15,
	legend columns=-1,
	legend entries={\strut L5, \strut L4, \strut L3, \strut L2},
	legend to name=CombinedLegendBar,
	ylabel={Percentage of Households},
	xlabel = {Low Demand Consumers},
	xlabel style={at={(0.5,-10ex)}},
	symbolic x coords={tool1, tool2, tool3, tool4, 
		tool5, tool6},
	xtick=data,
	xticklabels={$\alpha=0$,$\alpha=0.5$,$\alpha=1$,$\alpha=10000$,Equitable,Max Min},
	x tick label style={rotate=45,anchor=east},
	label style={font=\LARGE},
	tick label style={font=\LARGE},
	ymin=0,
	ymax=100,
	]
	\addplot+[ybar] plot coordinates {(tool1,0) (tool2,0) 
		(tool3,0) (tool4,0) (tool5,0) (tool6,0)};
	\addplot+[ybar] plot coordinates {(tool1,100) (tool2,100) 
		(tool3,100) (tool4,100) (tool5,100) (tool6,100)};
	\addplot+[ybar] plot coordinates {(tool1,0) (tool2,0) 
		(tool3,0) (tool4,0) (tool5,0) (tool6,0)};
	\addplot+[ybar] plot coordinates {(tool1,0) (tool2,0) 
		(tool3,0) (tool4,0) (tool5,0) (tool6,0)};
	\end{axis}
	\end{tikzpicture}%
	\begin{tikzpicture}[scale=0.40]
	\begin{axis}[
	ybar stacked,
	bar width=15pt,
	enlargelimits=0.15,
	legend columns=-1,
	legend entries={\strut L5, \strut L4, \strut L3, \strut L2},
	legend to name=CombinedLegendBar,
	xlabel = {Medium Demand Consumers},
	xlabel style={at={(0.5,-10ex)}},
	symbolic x coords={tool1, tool2, tool3, tool4, 
		tool5, tool6},
	xtick=data,
	xticklabels={$\alpha=0$,$\alpha=0.5$,$\alpha=1$,$\alpha=10000$,Equitable,Max Min},
	x tick label style={rotate=45,anchor=east},
	label style={font=\LARGE},
	tick label style={font=\LARGE},
	yticklabels={},
	ymin=0,
	ymax=100,
	]
	\addplot+[ybar] plot coordinates {(tool1,0) (tool2,0) 
		(tool3,0) (tool4,0) (tool5,0) (tool6,0)};
	\addplot+[ybar] plot coordinates {(tool1,100) (tool2,100) 
		(tool3,100) (tool4,100) (tool5,100) (tool6,100)};
	\addplot+[ybar] plot coordinates {(tool1,0) (tool2,0) 
		(tool3,0) (tool4,0) (tool5,0) (tool6,0)};
	\addplot+[ybar] plot coordinates {(tool1,0) (tool2,0) 
		(tool3,0) (tool4,0) (tool5,0) (tool6,0)};
	\end{axis}
	\end{tikzpicture}%
	\begin{tikzpicture}[scale=0.40]
	\begin{axis}[
	ybar stacked,
	bar width=15pt,
	enlargelimits=0.15,
	legend columns=-1,
	legend entries={\strut L5, \strut L4, \strut L3, \strut L2},
	legend to name=CombinedLegendBar,
	xlabel = {High Demand Consumers},
	xlabel style={at={(0.5,-10ex)}},
	symbolic x coords={tool1, tool2, tool3, tool4, 
		tool5, tool6},
	xtick=data,
	xticklabels={$\alpha=0$,$\alpha=0.5$,$\alpha=1$,$\alpha=10000$,Equitable,Max Min},
	x tick label style={rotate=45,anchor=east},
	label style={font=\LARGE},
	tick label style={font=\LARGE},
	yticklabels={},
	ymin=0,
	ymax=100,
	]
	\addplot+[ybar] plot coordinates {(tool1,0) (tool2,0) 
		(tool3,0) (tool4,0) (tool5,0) (tool6,0)};
	\addplot+[ybar] plot coordinates {(tool1,100) (tool2,94.764) 
		(tool3,94.764) (tool4,94.764) (tool5,94.764) (tool6,94.764)};
	\addplot+[ybar] plot coordinates {(tool1,0) (tool2,4.721)
		(tool3,4.721) (tool4,4.721) (tool5,4.721) (tool6,4.721)};
	\addplot+[ybar] plot coordinates {(tool1,0.52356) (tool2,0.52356)
		(tool3,0.52356) (tool4,0.52356) (tool5,0.52356) (tool6,0.52356)};
	\end{axis}
	\end{tikzpicture}
	\ref{CombinedLegendBar}
	\caption{Comparison of different categories of consumers in Sweden dataset with shortfall of $20 \%$ (top row) and $40 \%$ (bottom row). Figure is best seen in color}
	\label{fig_sweden_cust_category}
\end{figure}

Figure \ref{fig_comparison_aus} shows the comparison of the actual load with our proposed approach and the baselines for all households of Australia dataset on a randomly selected day (after aggregating the $24$ hours of that day). We can see that the allocation using the alpha fair approach is very close to the original load. The max min is the second best in terms of allocation, while the equitable give a straight line, which is the most inefficient approach. Same behavior is observed for other datasets as well (their results are not shown because of space constraints).
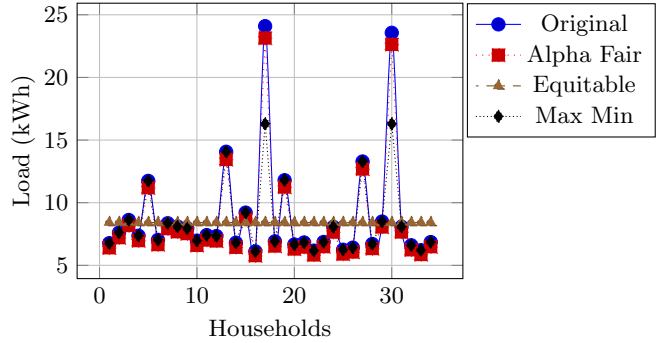
\begin{figure}[h!]
	\centering
	\begin{tikzpicture}[scale=1]
	\begin{axis}[title={},
	compat=newest,
	xlabel style={text width=3.5cm, align=center},
	xlabel={{\small Households}},
	ylabel={Load (kWh)}, 
	ylabel shift={-3pt},
	height=0.43\columnwidth, width=0.55\columnwidth, grid=major,
	legend columns=1,
	legend style ={ at={(1.5,1)}},
	legend entries={Original, Alpha Fair, Equitable, Max Min},
	]
	\addplot+[
	mark size=2.5pt,
	smooth,
	error bars/.cd,
	y fixed,
	y dir=both,
	y explicit
	] table [x={x}, y={Original}, col sep=comma] {actual_values_shortfall_5_percent_alpha_zero_aus.csv};
	\addplot+[
	mark size=2.5pt,
	dotted,
	mark=square*,
	error bars/.cd,
	y fixed,
	y dir=both,
	y explicit
	] table [x={x}, y={Alpha}, col sep=comma] {actual_values_shortfall_5_percent_alpha_zero_aus.csv};
	\addplot+[
	mark size=2.5pt,
	dash pattern=on 1pt off 3pt on 3pt off 3pt,
	mark=triangle*,
	error bars/.cd,
	y fixed,
	y dir=both,
	y explicit
	] table [x={x}, y={Equitable}, col sep=comma] {actual_values_shortfall_5_percent_alpha_zero_aus.csv};
	\addplot+[
	mark size=2.5pt,
	densely dotted,
	mark=diamond*,
	error bars/.cd,
	y fixed,
	y dir=both,
	y explicit
	] table [x={x}, y={MaxMin}, col sep=comma] {actual_values_shortfall_5_percent_alpha_zero_aus.csv};
	\end{axis}
	\end{tikzpicture}%
	\caption{Comparison of Actual load vs. the allocated load using different techniques for all households of Australia dataset for a randomly selected day. Figure is best seen in color}
	\label{fig_comparison_aus}
\end{figure}

The revenue generated for real world datasets and their comparison with the baselines are shown in Table \ref{table_revenue_australia}, \ref{table_revenue_sweden}, and \ref{table_revenue_irish}. It can be seen that the generated revenue increases as we increase the supply $S$. This behavior is obvious as large supply yields large revenue. As we increase the block rate threshold $b$ from $10^{th}$ to $90^{th}$ percentile, the revenue generated decreases. Since $p_2>p_1$ and as the value of $b$ increases, more units falls under $p_1$ block which produces less revenue. Also note that in comparison with the baselines (equitable and max-min), our approach generates more revenue while also fulfilling the fair distribution criteria. This high revenue generation highlights the effectiveness of our proposed approach.
\begin{table}[h!]
	\scriptsize
	\centering
	\begin{tabular}{@{\extracolsep{4pt}}p{1.2cm}p{0.67cm}p{0.67cm}p{0.67cm}p{0.67cm}p{0.67cm}p{0.67cm}p{0.67cm}p{0.67cm}p{0.67cm}p{0.67cm}p{0.67cm}p{0.67cm}p{0.67cm}@{}}
		
		\toprule
		\multirow{2}{*}{$\alpha$} &
		\multicolumn{4}{c}{$b=10^{th}$ } &
		\multicolumn{4}{c}{$b=50^{th}$ } & 
		\multicolumn{4}{c}{$b=90^{th}$ }
		\\
		\cline{2-5} \cline{6-9} \cline{10-13} 
		& {S=10\%} & {30\%} & {60\%} & {90\%} & {S=10\%} & {30\%} & {60\%} & {90\%} & {S=10\%} & {30\%} & {60\%} & {90\%} \\
		\midrule
		
		0 & 3512.8 & 3280.8 & 2805.9 & 1946.4 & 3395.7 & 3186.3 & 2753.6 & 1914.2 & 3151.9 & 2971.3 & 2582.5 & 1839.2 \\
		0.5 & 3467.6 & 3166.3 & 2563.6 & 1808 & 3323.1 & 3021.8 & 2419.1 & 1808 & 3043.1 & 2741.8 & 2410.6 & 1808 \\
		1 & 3467.6 & 3166.3 & 2563.6 & 1808 & 3323.1 & 3021.8 & 2419.1 & 1808 & 3043.1 & 2741.8 & 2410.6 & 1808 \\
		10000 & 3509.6 & 3232.2 & 2669.7 & 1808 & 3392 & 3124.2 & 2586 & 1808 & 3147.6 & 2902.6 & 2424.9 & 1808 \\
		\hline
		Equitable & 3441.5 & 	3140.1 & 	2537.5 & 	1808 & 	3188 & 	2886.7 & 	2410.6 & 	1808 & 	2862.6 & 	2712 & 	2410.6 & 	1808 \\
		Max-Min  & 3467.6 & 	3166.3 & 	2563.6	 & 1808	 & 3323.1 & 	3021.8 & 	2419.1 & 	1808 & 	3043.1	 & 2741.8 & 	2410.6	 & 1808\\
		
		\bottomrule
	\end{tabular}
	\caption{Revenue generated using Australia dataset with changing value of $\alpha$, shortfall, and  block threshold $b$}
	\label{table_revenue_australia}
\end{table}
\begin{table}[h!]
	\scriptsize
	\centering
	\begin{tabular}{@{\extracolsep{4pt}}p{1.2cm}p{0.67cm}p{0.67cm}p{0.67cm}p{0.67cm}p{0.67cm}p{0.67cm}p{0.67cm}p{0.67cm}p{0.67cm}p{0.67cm}p{0.67cm}p{0.67cm}p{0.67cm}@{}}
		
		\toprule
		\multirow{2}{*}{$\alpha$} &
		\multicolumn{4}{c}{$b=10^{th}$ } &
		\multicolumn{4}{c}{$b=50^{th}$ } & 
		\multicolumn{4}{c}{$b=90^{th}$ }
		\\
		\cline{2-5} \cline{6-9} \cline{10-13} 
		& {S=10\%} & {30\%} & {60\%} & {90\%} & {S=10\%} & {30\%} & {60\%} & {90\%} & {S=10\%} & {30\%} & {60\%} & {90\%} \\
		\midrule
		
		0 & 155920 & 145280 & 123990 & 81626 & 134180 & 123980 & 104020 & 68307 & 111120 & 103640 & 89662 & 64696 \\
		0.5 & 155890 & 145210 & 123840 & 81110 & 133780 & 123100 & 101730 & 64096 & 107990 & 97303 & 85462 & 64096 \\
		1 & 155890 & 145210 & 123840 & 81110 & 133780 & 123100 & 101730 & 64096 & 107990 & 97303 & 85462 & 64096 \\
		10000 & 155900 & 145230 & 123870 & 81169 & 134070 & 123530 & 102460 & 64207 & 110840 & 102040 & 86001 & 64096 \\
		\hline
		Equitable & 155830 & 145150 & 123780 & 81050 & 133130 & 122450 & 101080 & 64096 & 101490 & 96144 & 85462 & 64096 
		\\
		Max-Min & 155890 & 145210 & 123840 & 81110 & 133780 & 123100 & 101730 & 64096 & 107990 & 97303 & 85462 & 64096 \\
		\bottomrule
	\end{tabular}
	\caption{Revenue generated using Sweden dataset with changing value of $\alpha$, shortfall, and  block threshold $b$}
	\label{table_revenue_sweden}
\end{table}
Recall that increasing $\alpha$ favors smaller allocations. We can notice this effect as we increase $\alpha$, the revenue generated decreases. However, due to single source, the proportional fair ($\alpha=1$) is equal to max-min fair ($\alpha=\infty$) \cite{kelly1997charging}. 
The proportional fair gives priority to smaller allocation but less emphatically. As we follow block rate pricing, favoring smaller allocation will not give a high profit. That is why the revenue generated for $\alpha=1$ is less than the revenue for $\alpha=1000$ in most cases. The maximum revenue is generated for $\alpha=0$.

\begin{table}[h!]
	\scriptsize
	\centering
	\begin{tabular}{@{\extracolsep{4pt}}p{1.2cm}p{0.67cm}p{0.67cm}p{0.67cm}p{0.67cm}p{0.67cm}p{0.67cm}p{0.67cm}p{0.67cm}p{0.67cm}p{0.67cm}p{0.67cm}p{0.67cm}p{0.67cm}@{}}
		
		\toprule
		\multirow{2}{*}{$\alpha$} &
		\multicolumn{4}{c}{$b=10^{th}$ } &
		\multicolumn{4}{c}{$b=50^{th}$ } & 
		\multicolumn{4}{c}{$b=90^{th}$ }
		\\
		\cline{2-5} \cline{6-9} \cline{10-13} 
		& {S=10\%} & {30\%} & {60\%} & {90\%} & {S=10\%} & {30\%} & {60\%} & {90\%} & {S=10\%} & {30\%} & {60\%} & {90\%} \\
		\midrule
		
		0 & 77210 & 71284 & 60283 & 40917 & 75077 & 69516 & 59102 & 40532 & 67065 & 62959 & 54875 & 39301 \\
		0.5 & 77034 & 70687 & 57994 & 38080 & 74601 & 68254 & 55561 & 38080 & 64775 & 58429 & 50773 & 38080 \\
		1 & 77034 & 70687 & 57994 & 38080 & 74601 & 68254 & 55561 & 38080 & 64775 & 58429 & 50773 & 38080 \\
		10000 & 77166 & 70919 & 58529 & 38083 & 75000 & 68870 & 56808 & 38080 & 66947 & 61566 & 51442 & 38080 \\
		\hline
		Equitable & 77025 & 70679 & 57985 & 38080 & 74430 & 68084 & 55390 & 38080 & 60293 & 57120 & 50773 & 38080 \\
		Max-Min & 77034 & 70687 & 57994 & 38080 & 74601 & 68254 & 55561 & 38080 & 64775 & 58429 & 50773 & 38080 \\
		\bottomrule
	\end{tabular}
	\caption{Revenue generated using Ireland dataset with changing value of $\alpha$, shortfall, and  block threshold $b$}
	\label{table_revenue_irish}
\end{table}

\section{Conclusion and Future Work} \label{Section_Conclusion}
In this paper, we propose an alternative to the load shedding scheme called SLS. We mathematically formulate the problem of SLS as a social welfare optimization problem. Solving the SLS problem is equivalent to maximizing the net utility of each consumer. We introduce a parametric notion of fairness called $\alpha$-fair, where $\alpha \in [0,\infty]$. As we increase $\alpha$, smaller demands are given more preference. When the total supply $S$ is too large or too small, different values of $alpha$ turns in similar allocations. For single supply unit, proportional fair ($\alpha = 1$)  is equal to max-min fair ($\alpha = \infty$). We consider block rate pricing to study the effect of fairness on revenue generation. As we increase $\alpha$, the generated revenue decreases. The decision-makers can tune the value of $\alpha$ according to different scenarios of fairness and revenue generation. We compare our approach with two baselines and show that our proposed method is better in terms of fair allocation and revenue generation.
As the utility companies have to deal with multi-million customers, one of the future direction is to compute the closed-form solution of the problem or at least a parallel solution. Another future direction is to formulate the SLS problem, which can incorporate the distributed electricity generations of the smart grid.

\bibliographystyle{unsrt}
\bibliography{Fair_Allocation_Based_Soft_Load_Shedding}

\end{document}